\documentclass[14]{article}
\usepackage{arxiv}
\usepackage{amsmath}

\usepackage[utf8]{inputenc} 
\usepackage[T1]{fontenc}    
\usepackage{hyperref}       
\usepackage{url}            
\usepackage{booktabs}       
\usepackage{amsfonts}       
\usepackage{nicefrac}       
\usepackage{microtype}      
\usepackage{lipsum}
\usepackage{cases}
 \usepackage{setspace}
\usepackage{geometry}
\usepackage{graphicx}
\usepackage{subcaption}%
\usepackage{caption}
\usepackage{tikz, pgfplots}
\usepackage{tikz}
\usetikzlibrary{3d}
\usetikzlibrary{hobby}
\usepackage{amssymb}
\usepackage{mathrsfs}
\usepackage{latexsym}
\usepackage{enumerate}
\usepackage{color}
\usepackage{ulem}
\usepackage{amsfonts}
\usepackage{verbatim}
\usepackage{amsthm}
\usepackage{dsfont}
\usepackage{url}
\usepackage{tocloft}
\usepackage{youngtab} 
 \usepackage{ytableau} 
\usepackage{enumerate}
\usepackage{mathtools}
\usepackage{multirow}
\usepackage{array}
\usepackage{xpatch}
\allowdisplaybreaks
\usepackage{cellspace} 
\pgfplotsset{compat=1.18}

\xapptocmd{\proof}{\mbox{}\par\nobreak}{}{}
\newtheorem{theorem}{Theorem}[section]
\newtheorem{assumptions}[theorem]{Assumptions}

\newtheorem{lemma}[theorem]{Lemma}

\newtheorem{definition}[theorem]{Definition}
\newtheorem{prop}[theorem]{Proposition}

\numberwithin{equation}{section}

\newcolumntype{C}[1]{>{\centering\arraybackslash}p{#1}}

\title{Kernel Smoothing for Bounded Copula Densities}

\author{
   Mathias Muia\thanks{Department of Mathematics and Statistics, University of South Alabama, 411 University Boulevard North, Mobile, AL 36688. Email: \texttt{mnmuia@southalabama.edu}} \\
   \And
   Olivia Atutey\thanks{Department of Mathematics and Statistics, University of South Alabama, 411 University Boulevard North, Mobile, AL 36688. Email: \texttt{oatutey@southalabama.edu}}
   \\
   \And
   Mahmud Hasan\thanks{Department of Biostatistics, Virginia Commonwealth University, Richmond, VA, USA. Email: \texttt{hasanm10@vcu.edu}}
}

\renewcommand{\arraystretch}{1.5}  
\begin{document}
\maketitle

\begin{abstract}
Nonparametric estimation of copula density functions using kernel estimators presents significant challenges. One issue is the potential unboundedness of certain copula density functions at the corners of the unit square. Another is the boundary bias inherent in kernel density estimation. This paper presents a kernel-based method for estimating bounded copula density functions, addressing boundary bias through the mirror-reflection technique. Optimal smoothing parameters are derived via Asymptotic Mean Integrated Squared Error (AMISE) minimization and cross-validation, with theoretical guarantees of consistency and asymptotic normality. Two kernel smoothing strategies are proposed: the rule-of-thumb approach and least squares cross-validation (LSCV). Simulation studies highlight the efficacy of the rule-of-thumb method in bandwidth selection for copulas with unbounded marginal supports. The methodology is further validated through an application to the Wisconsin Breast Cancer Diagnostic Dataset (WBCDD), where LSCV is used for bandwidth selection.

\end{abstract}

\section{Introduction}
The study of dependence between random variables is a mainstay of statistical analysis. It is a critical, yet challenging task in multivariate statistical modeling, as it requires specifying complex joint distributions of random variables to fully capture their dependence structure. This complexity can be overcome by using copula models, which disentangle the marginal distributions from the dependence structure of the joint distribution. A multivariate distribution can be fully characterized by its marginal distributions and an associated copula, making copulas an indispensable tool for statistical investigations (Darsow \textit{et al}. (1992) \cite{darsow1992copulas}; Nelsen (2006) \cite{nelsen2006methods}). For further details on copulas, see Joe (1997) \cite{joe1997multivariate}, Nelsen (2006) \cite{nelsen2006methods}, and Durante and Sempi (2015) \cite{durante2015principles}. In what follows, we focus on the bivariate case only for simplicity, the results being extendable to more than two dimensions.

Sklar's theorem (Sklar, 1959 \cite{sklar1959}) is central to the theoretical foundation of copulas. It states that if \((X, Y)\) is a pair of random variables with a joint distribution function \(H\) and continuous marginal distribution functions \(F_1\) and \(F_2\), then there exists a unique copula \(C\) such that for all \(x,y, \ H(x, y) = C(F_1(x), F_2(y)).\) The copula \(C\) fully describes the dependence structure between \( X \) and \( Y \).

Formally, a function \(C: [0,1]^{2}\rightarrow [0,1]\) is called bivariate copula if it
satisfies the following conditions:\begin{itemize}
    \item[i.]  $C(0,u)=C(u,0)=0, \quad C(u,1)=C(1,u)=u, \quad \forall \ u\in [0,1];$
    \item[ii.] $C(u_{1},v_{1})+C(u_{2},v_{2})-C(u_{1},v_{2})-C(u_{2},v_{1})\geq 0, \quad \forall\ [u_{1},u_{2}]\times[v_{1},v_{2}]\subset [0,1]^{2}.$
\end{itemize}

If \( C \) is absolutely continuous (see Nelsen (2006) \cite{nelsen2006methods}), then it admits a joint density given by:
\[
    c(u,v) = \frac{\partial^2 C}{\partial u \partial v} = \frac{\partial^2 C}{\partial v \partial u}.
\]
In this case, such a copula is said to possess a density. The primary objective of this paper is to propose an estimator for \( c \), given an i.i.d. sample \(\{ (X_i,Y_i)\}_{i=1}^n\) from the distribution function \( H \), and analyze its properties.

Methods for estimating copulas and copula densities generally depend on the assumptions made about the joint distribution function \( H \). Commonly used approaches include parametric methods (Iyengar \textit{et al.}, 2011 \cite{iyengar2011}), semiparametric methods (Chen and Fan, 2006 \cite{chen2006}), and nonparametric methods (Chen and Huang, 2007 \cite{chen2007}; Wang \textit{et al.}, 2012 \cite{wang2012}). In fully parametric models, where both the copula and the marginal distributions are explicitly specified, maximum likelihood estimation is the method of choice. Semiparametric approaches, on the other hand, often specify a parametric copula while estimating the marginals nonparametrically, providing greater flexibility. 

Striking a balance between accuracy and computational efficiency is essential in copula estimation. Trivedi and Zimmer (2005) \cite{trivedi2005} and Choros \textit{et al.} (2010) \cite{choros2010} provided a discussion of these trade-offs. However, parametric methods can suffer significant underestimation when the marginal distributions are unknown or misspecified, as noted by Charpentier \textit{et al.} (2007) \cite{charpentier2007}. In such cases, non-parametric methods emerge as a more robust and flexible alternative. For example, Behnen \textit{ et al.} (1985) \cite{behnen1985rank} proposed using a rank-based estimator, which modifies kernel estimators for application to rank data, to estimate the copula density. Later, Gijbels and Mielniczuk (1990) \cite{gijbels1990} introduced a kernel-type estimator for bivariate copula densities and demonstrated its consistency and asymptotic normality under various conditions related to bandwidth selection and kernel smoothness.

These advancements underscore the effectiveness of nonparametric methods in capturing complex dependence structures without imposing restrictive assumptions on the copula or its marginals (Charpentier \textit{et al.} (2007) \cite{charpentier2007}; Chen and Huang, 2007 \cite{chen2007}). Building on these techniques, this paper employs a two-stage kernel density estimation procedure to estimate the copula density. In the first stage, the marginal distribution functions \( F_1(X) \) and \( F_2(Y) \) are estimated nonparametrically using empirical CDFs. These marginal estimates are then utilized in the second stage to construct a flexible, non-parametric estimate of the copula density. This approach provides a robust framework for modeling dependence structures, particularly when the underlying distribution functions are complex or unknown.


To formalize the methods, we introduce several notations that will be used throughout the paper. The squared \( L_2 \) norm of a function \( g \) is defined as
\begin{align}\label{l2norm}
    R(g) \equiv \|g\|_2^2 = \int_{-\infty}^{\infty} g(x)^2 \, dx,
\end{align}
where \( R(g) \) can be interpreted as a measure of the roughness of \( g \). The square of the \( p^{\text{th}} \) derivative of \( g \) is denoted by \( g^{(p)}(x)^2 \). Moments of the kernel \( K \) are represented by \( \mu_k \), where
\begin{align}\label{moments dfn}
    \mu_k = \int s^k K(s) \, ds.
\end{align}
A symmetric kernel function is characterized by the property \( k(u) = k(-u) \) for all \( u \). In this case, all odd moments are equal to zero. Most nonparametric estimation techniques utilize symmetric kernels, and our discussion will focus on them.
For a kernel \( K \), we assume that \( \mu_0 = 1 \), \( \mu_1 = \cdots = \mu_{p-1} = 0 \), and \( 0 < |\mu_p| < \infty \) for some even \( p \). These moment conditions play a pivotal role in determining the theoretical properties of kernel-based estimators. Unless otherwise stated, integrals are taken over the entire real line.

The remainder of the paper is organized as follows. Section~\ref{estimating the copula} introduces the proposed estimator. In Section~\ref{properties of the estimator}, we examine the properties of the estimator, including its bias and asymptotic behavior. Section~\ref{smoothing} discusses bandwidth selection procedures, detailing two main methods: the rule-of-thumb approach and least squares cross-validation (LSCV). Section~\ref{simulation} presents a simulation study using data generated from the Frank copula to validate our methodology. Finally, the results are applied to the Wisconsin Breast Cancer dataset in Section~\ref{real-life}.

\section{Estimating the Copula Density}\label{estimating the copula}

Let \(\{ (X_i,Y_i)\}_{i=1}^n\) be an i.i.d. sample from the joint distribution function \( H \). To estimate the copula density, we consider the transformed sample \(\{(\hat{F}_1(X_i), \hat{F}_2(Y_i))\}_{i=1}^n\), where \(\hat{F}_1\) and \(\hat{F}_2\) are the empirical distribution functions of the marginal distributions. Using empirical estimators for the margins has distinct advantages and trade-offs. On the one hand, this approach introduces zero bias into the copula density estimator and provides an exact representation of the data, unlike kernel-based margin estimators as discussed in Chen and Huang (2007) \cite{chen2007}. On the other hand, it results in higher variance, which may affect the overall accuracy of the estimator in the case of small samples.

Based on the transformed sample, the natural kernel-type estimator of $c(u, v)$ is (Wand and Jones (1995) \cite{wand1995kernel}; Silverman (2018) \cite{silverman2018density})  
\begin{align}\label{estimator1}
    \frac{1}{nh_n^2} \sum_{i=1}^n K \left( \frac{u - \hat{F}_1(X_i)}{h_n}, \frac{v - \hat{F}_2(Y_i)}{h_n} \right),
\end{align}
where \( K \) is a bivariate density function that is symmetric, unimodal at \((0,0)\), and has support on \([-1, 1]\). The parameter \( h_n \) represents a sequence of bandwidths that converges to zero as the sample size increases. Commonly used kernel functions include the uniform, Epanechnikov, biweight, triweight, and Gaussian kernels, each of which affects the smoothness and accuracy of the estimate differently. Detailed discussions on kernel density estimation can be found in Wand and Jones (1995) \cite{wand1995kernel}, Bowman \textit{et al.} (1998) \cite{bowman1998bandwidth}, and Silverman (2018) \cite{silverman2018density}.

However, the estimator (\ref{estimator1}) has a notable limitation. It is inconsistent at points on the boundary of the unit square in which $c$ has a jump, as noted by Gijbels and Mielniczuk (1990) \cite{gijbels1990}. This arises because, near the boundary, the summands in (\ref{estimator1}) allocate a significant amount of mass outside the unit square, leading to an issue known as boundary bias. Three methods have been proposed to address this problem, all of which were initially developed in the context of univariate kernel density estimation on the unit line. The first approach employs mirror reflection techniques, as discussed by Deheuvels and Hominal (1979) \cite{deheuvels1979}, Schuster (1985) \cite{schuster1985}, and Behnen \textit{et al.} (1985) \cite{behnen1985rank}. In two dimensions, this method reflects data points across all edges and corners of the unit square, generating an expanded dataset from which the kernel estimate is constructed. By redistributing the kernel’s mass within the unit square, this approach mitigates boundary bias. The second method uses a local linear variant of the kernel estimator, as introduced by Chen and Huang (2007) \cite{chen2007}, to address the bias near boundaries. The third method involves employing kernels whose support aligns with the support of the target density. These kernels adjust their shape based on the location where the density is being estimated. They are known as \textit{boundary kernels} and are exemplified by the \textit{beta kernel}, proposed by Charpentier \textit{et al.} (2007) \cite{charpentier2007}.

In this work, we employ the mirror reflection technique, which yields the following estimate of \( c(u,v) \). Denoting \(\hat{F}_1(X_i)\) and \(\hat{F}_2(Y_i)\) by \(\hat{U}_i\) and \(\hat{V}_i\), respectively, we define the estimator for \( c(u,v) \) as follows (Behnen \textit{et al.} (1985) \cite{behnen1985rank}):

\begin{definition}\label{definition1}
Let \( c(u,v) \) be a copula density. The \textbf{mirror-reflection kernel estimate of \( c \)} with smoothing parameter \( h_n > 0 \) is defined by
\begin{align}\label{modified estimate of c}
    \hat{c}(u,v) =\begin{cases}
        \frac{1}{nh_n^2} \sum\limits_{i=1}^n\sum\limits_{l=1}^9 K \left( \frac{u - \hat{U}_{il}}{h_n}, \frac{v - \hat{V}_{il}}{h_n} \right), & \text{for}\ 0 \leq u, v \leq 1,\\
        0, & \text{otherwise},
    \end{cases}
\end{align}
where \(\lbrace (\hat{U}_{il}, \hat{V}_{il}), i = 1, \dots, n; \, l = 1, \dots, 9 \rbrace = \lbrace (\pm \hat{U}_i, \pm \hat{V}_i), (\pm \hat{U}_i, 2 - \hat{V}_i), (2 - \hat{U}_i, \pm \hat{V}_i), (2 - \hat{U}_i, 2 - \hat{V}_i), i = 1, \dots, n \rbrace.\) 
\end{definition}

The following conditions are assumed in the analysis\begin{assumptions}\label{assumptions}
  \begin{itemize}
      \item[\textbf{$A_1$}:] $H$ has continuous marginal distribution functions and the copula density $c$ has bounded derivatives up to the fourth order on $[0,1]^2;$ 
      \item[\textbf{$A_2$}:]$K$ is a symmetric, bounded, sufficiently smooth continuous probability density supported on $[-1,1]$;
      \item[\textbf{$A_3$}:] The bandwidths satisfy $h_n\longrightarrow 0$ as $n\longrightarrow\infty$ and $ nh_n^2\longrightarrow\infty$ as $n\longrightarrow\infty.$ 
  \end{itemize}  
\end{assumptions}

Some widely applied copula densities, such as the Clayton, Gumbel, Gaussian, and Student’s $t$ copulas, as well as their first- and second-order partial derivatives, exhibit unboundedness at the boundaries of the unit square. This behavior introduces complexities in the analysis of these derivatives near the corners and edges of the unit square. To address this challenge in their work on copula estimation, Omelka \textit{\textit{et al.}} (2009) \cite{omelka2009improved} used the ``bandwidth shrinking'' method when approaching the borders of $[0,1]^2$. This approach involves substituting the bandwidth \( h_n \) with a bandwidth function \( r(w)h_n \), where \( w = u \) or \( v \), effectively ``shrinking'' the bandwidth towards zero at the corners of the unit square.

In this work, we focus on the estimation of copula densities that are bounded and possess bounded second-order partial derivatives, exemplified by families such as Ali-Mikhail-Haq, Frank, Plackett, Farlie-Gumbel-Morgenstern, among others.

\section{Properties of the estimator}\label{properties of the estimator}

In this section, we analyze the properties of the estimator defined in \eqref{modified estimate of c}. For this purpose, we utilize the multiplicative kernel, which is specified as follows:  
\begin{align}\label{mult kernel}
    \hat{c}(u, v) = \frac{1}{n} \sum_{i=1}^n \sum_{l=1}^9 K_{h_n} \left( u - \hat{U}_{il} \right) K_{h_n} \left( v - \hat{V}_{il} \right),
\end{align}
where \( K_{h_n}(x) = 1/h_n K(x/h_n)\). This kernel function will be employed consistently throughout the remainder of the paper.

\subsection{Bias and variance}
The following proposition establishes the expressions for the estimator's bias and variance.

\begin{prop}\label{prop:bias_variance}
Let \( c(u,v) \) be a copula density function that is twice-continuously differentiable on \( [0,1]^2 \). Assume that the smoothing parameter \( h_n \to 0 \) and \( nh_n^2 \to \infty \) as \( n \to \infty \). Then, for all points \( (u,v) \in [0,1]^2 \), the bias and variance of the estimator \( \hat{c}(u,v) \), defined in (\ref{mult kernel}), are given by:
\begin{align*}
    &\operatorname{Bias}[\hat{c}(u,v)] = \frac{\mu_2(K) h_n^2}{2} \left[ c_{uu}(u,v) + c_{vv}(u,v) \right] + o(h_n^2), \\
    &\operatorname{Var}[\hat{c}(u,v)] = \frac{R(K)^2}{nh_n^2} c(u,v) + o\left(\frac{1}{nh_n^2}\right),
\end{align*}
where \( R(K) \) and \( \mu_2(K) \) are as defined in \eqref{l2norm} and \eqref{moments dfn} respectively. Here, \( c_{uu} \) and \( c_{vv} \) denote the second-order partial derivatives of \( c \) with respect to \( u \) and \( v \), respectively.
\end{prop}

\begin{proof}\begin{itemize}
    \item[i.] The bias is given by  \begin{align}\label{biasproof}
    Bias[\hat{c}(u,v)]=\sum\limits_{l=1}^9 E\left[K_{h_n}(u-\hat{U}_{il})K_{h_n}(v-\hat{V}_{il})\right]-c(u,v)
\end{align}   
To show that this bias is uniformly $O(h_n^2)$, we utilize the Taylor expansion of \( c(u - h_n s, v - h_n t) \). This expansion is justified by the assumption that \( c(u, v) \) has bounded first and second-order derivatives on \( [0,1]^2 \). 

The simplest case occurs when \( (u, v) \in [h_n, 1 - h_n]^2 \). In this case, Equation (\ref{biasproof}) reduces to
\begin{align*}
Bias[\hat{c}(u,v)]&=  E\left[K_{h_n}(u-\hat{U}_1)K_{h_n}(v-\hat{V}_1)\right]-c(u,v).  
\end{align*}
Substituting $s=(u-x)/h_n$ and $t=(v-y)/h_n$ and using Assumptions \ref{assumptions}, we obtain
\begin{align*}
Bias[\hat{c}(u,v)]&=  \int_{-1}^1\int_{-1}^1 c(u-h_ns,v-h_nt)K(s)K(t)dsdt -c(u,v)  \\
&=\int_{-1}^1\int_{-1}^1\lbrace c(u,v)-h_nsc_u(u,v)-h_ntc_v(u,v)+\frac{h_n^2s^2}{2}c_{uu}(u,v) + \\ &\qquad h_n^2stc_{uv}(u,v) +\frac{h_n^2t^2}{2}c_{vv}(u,v)\rbrace K(s)K(t)dsdt+o(h_n^2)-c(u,v)\\
&=\frac{\mu_2(K)h_n^2}{2}\left[ c_{uu}(u,v)+c_{vv}(u,v) \right] +o(h_n^2)
\end{align*}
where $\mu_2(K)$ denotes the second moment of $K$, and $c_u$, $c_v$, $c_{uu}$, $c_{uv}$, and $c_{vv}$ are the first- and second-order partial derivatives with respect to $u$ and $v$. 

Regarding the remaining cases, we will show only for the case when $(u,v)\in [1-h_n,1]^2$. The other cases may be handled in a similar way. Note that the Taylor expansion together with Assumptions \ref{assumptions} imply that  $c(u,v)=c(1,1)+(u-sh_n)c_u(1,1)+(v-th_n)c_v(1,1)+O(h_n^2)$ uniformly in $(u,v)\in[1-2h_n,1]^2.$ Further routine algebra shows that (\ref{biasproof}) simplifies to\begin{align*}
    Bias[\hat{c}(u,v)]&=EK_{h_n}(u-U_1)K_{h_n}(v-V_1)\\
    &\qquad +EK_{h_n}(u-U_1)K_{h_n}(v+V_1-2)\\
    &\qquad +EK_{h_n}(u+U_1-2)K_{h_n}(v-V_1)\\
    &\qquad +EK_{h_n}(u+U_1-2)K_{h_n}(v+V_1-2)-c(u,v)
\end{align*}
Let $G(z)=\int_{-\infty}^zK(t)dt$ and $T_{w,h_n}=G((w-1)/h)$ for $w=u\ \text{or}\ v$. We compute \begin{align}\label{1}
    E&K_{h_n}(u-U_1)K_{h_n}(v-V_1)\nonumber\\
    &\qquad=\int\int K_{h_n}(u-x)K_{h_n}(v-y)c(x,y)dxdy\nonumber\\
    &\qquad=\int_{\frac{v-1}{h_n}}^1\int_{\frac{u-1}{h_n}}^1 c(u-sh_n,v-th_n)K(s)K(t)dsdt\nonumber\\
    &\qquad=\int_{\frac{v-1}{h_n}}^1\int_{\frac{u-1}{h_n}}^1\lbrace c(1,1)+uc_u(1,1)-sh_nc_u(1,1)-c_u(1,1)+vc_v(1,1)-th_nc_v(1,1)-c_v(1,1)\rbrace\nonumber\\ &\qquad\qquad K(s)K(t)dsdt + O(h_n^2)\nonumber\\
    &\qquad=c(1,1)\left[1-T_{v,h_n}\right]\left[1-T_{u,h_n}\right]+uc_u(1,1)\left[1-T_{v,h_n}\right]\left[1-T_{u,h_n}\right]-c_u(1,1)\left[1-T_{v,h_n}\right]\times\nonumber\\&\qquad\qquad\left[1-T_{u,h_n}\right]-  h_nc_u(1,1)\left[1-T_{v,h_n}\right]\int_{\frac{u-1}{h_n}}^1sK(s)ds    +vc_v(1,1)\left[1-T_{v,h_n}\right]\left[1-T_{u,h_n}\right]-\nonumber\\&\qquad\qquad c_v(1,1)\left[1-T_{v,h_n}\right]\left[1-T_{u,h_n}\right]-h_nc_v(1,1)\left[1-T_{u,h_n}\right]\int_{\frac{v-1}{h_n}}^1tK(t)dt+O(h_n^2).
\end{align}
Similarly,
\begin{align}\label{2}
    E&K_{h_n}(u-U_1)K_{h_n}(v+V_1-2) 
   =\int\int K_{h_n}(u-x)K_{h_n}(v+y-2)c(x,y)dxdy\nonumber\\
    &\qquad= \int_{-1}^{\frac{v-1}{h_n}}\int_{\frac{u-1}{h_n}}^1 c(u-sh_n,2+th_n-v)K(s)K(t)dsdt\nonumber\\
    &\qquad=c(1,1)[1-T_{u,h_n}]T_{v,h_n}+uc_u(1,1)[1-T_{u,h_n}]T_{v,h_n}-h_nc_u(1,1)\int_{\frac{u-1}{h_n}}^1sK(s)dsT_{v,h_n}-\nonumber\\&\qquad\qquad c_u(1,1)[1-T_{u,h_n}]T_{v,h_n}+c_v(1,1)[1-T_{u,h_n}]T_{v,h_n}+h_nc_v(1,1)\int_{-1}^{\frac{v-1}{h_n}}tK(t)dt[1-T_{u,h_n}]-\nonumber\\&\qquad\qquad vc_v(1,1)[1-T_{u,h_n}]T_{v,h_n}+O(h_n^2).
\end{align}
Taking a similar approach as above gives
\begin{align}\label{3}
    E&K_{h_n}(u+U_1-2)K_{h_n}(v-V_1)=\int\int K_{h_n}(u+x-2)K_{h_n}(v-y) c(x,y)dxdy\nonumber\\
    &\qquad=c(1,1)[1-T_{v,h_n}]T_{u,h_n}- uc_u(1,1)[1-T_{v,h_n}]T_{u,h_n}+h_nc_u(1,1)\int_{-1}^{\frac{u-1}{h_n}}sK(s)ds[1-T_{v,h_n}]+\nonumber\\&\qquad\qquad c_u(1,1)[1-T_{v,h_n}]T_{u,h_n}-c_v(1,1)[1-T_{v,h_n}]T_{u,h_n}-h_nc_v(1,1)\int_{\frac{v-1}{h_n}}^1tK(t)dtT_{u,h_n}+ \nonumber\\&\qquad\qquad vc_v(1,1)[1-T_{v,h_n}]T_{u,h_n}+O(h_n^2).
\end{align}
and 
\begin{align}\label{4}
    E&K_{h_n}(u+U_1-2)K_{h_n}(v+V_1-2) 
   =\int\int K_{h_n}(u+x-2)K_{h_n}(v+y-2)c(x,y)dxdy\nonumber\\
    &=c(1,1)T_{v,h_n}T_{u,h_n}+c_u(1,1)T_{v,h_n}T_{u,h_n}+h_nc_u(1,1)\int_{-1}^{\frac{u-1}{h_n}}sK(s)dsT_{v,h_n}-uc_u(1,1)T_{v,h_n}T_{u,h_n}+\nonumber\\ &\qquad\qquad c_v(1,1)T_{v,h_n}T_{u,h_n}+h_nc_v(1,1)+\int_{-1}^{\frac{v-1}{h_n}}tK(t)dtT_{u,h_n}-vc_v(1,1)T_{v,h_n}T_{u,h_n}+O(h_n^2)
\end{align}

Combining (\ref{1}) with (\ref{2}), (\ref{3}) and (\ref{4}) gives us
\begin{align*}
    Bias[\hat{c}(u,v)]=c(1,1)+(u-sh_n)c_u(1,1)+(v-th_n)c_v(1,1)+O(h_n^2)-c(u,v)=O(h_n^2),
\end{align*}
which was to be proved.

\item[ii.] The prove for the variance relies on results on the asymptotic normality for the mirror-reflection estimator by  Gijbels and Mielniczuk (1990) \cite{gijbels1990}, Theorem 3.2. 
\end{itemize}
\end{proof}

\subsection{Consistency and asymptotic normality of the estimator}
Behnen \textit{et al.} (1985) \cite{behnen1985rank}, in Theorem 2.1, showed that the estimator (\ref{modified estimate of c}) is a uniformly strongly consistent estimator of $c(u,v)$. The asymptotic normality of the estimator is established in Theorem 3.2 of Gijbels and Mielniczuk (1990) \cite{gijbels1990}.

\section{Bandwidth selection}\label{smoothing}
The estimation process is highly sensitive to the choice of smoothing parameters, \( h_n \). In this section, we discuss three main methods for selecting bandwidth. Determining the ``optimal'' values for these parameters requires a specific criterion function.

The first method identifies the optimal bandwidth by minimizing the AMISE. The other two methods involve cross-validation (CV) approaches aimed at minimizing the MISE, following the works of Bowman (1984) \cite{bowman1984alternative}, Scott and Terrell (1987) \cite{scott1987biased}, and Terrell and Scott (1985) \cite{terrell1985oversmoothed}. In the context of multivariate kernel estimation, these data-driven CV methods seek to estimate the Integrated Squared Error (ISE) or MISE directly from the data, identifying the smoothing parameter that minimizes the estimated ISE or MISE.

Stone (1984) \cite{stone1984asymptotically} demonstrated a significant theoretical finding that, when the underlying multivariate density and its one-dimensional marginals are bounded, the smoothing parameters selected using multivariate least-squares cross-validation (LSCV) are asymptotically optimal. For more information on cross-validation methods, refer to Rudemo (1982) \cite{rudemo1982empirical}, Bowman (1984) \cite{bowman1984alternative}, Scott and Terrell (1987) \cite{scott1987biased}, and Terrell and Scott (1985) \cite{terrell1985oversmoothed}. Higher-order plug-in algorithms for multivariate density estimation have been explored by Wand and Jones (1995) \cite{wand1995kernel}.

\subsection{Rule-of-thumb approach}
A common method for selecting bandwidth involves establishing a rule-of-thumb based on a specific parametric family. Given that our interest lies in estimating density over the entire unit square, our focus is on choosing a global bandwidth. Consequently, we adopt the strategy of minimizing the asymptotic mean integrated squared error (AMISE).

The mean square error (MSE) for $(u,v)\in[h_n,1-h_n]^2$ is \begin{align}\label{MSE}
    MSE(\hat{c}(u,v))&=Var(\hat{c}(u,v))+Bias(\hat{c}(u,v))^2\nonumber\\
    &=\frac{R(K)^2}{nh_n^2}c(u,v)+[\frac{\mu_2(K)h_n^2}{2}[c_{uu}(u,v)+c_{vv}(u,v)]]^2+o((nh_n^2)^{-1}+h^4)
\end{align}
The MISE is obtained by integrating (\ref{MSE}) above using the integrability assumption on $c(u,v)$.
\begin{align}\label{MISE}
   MI&SE(\hat{c}(u,v))\nonumber\\= &(nh_n^2)^{-1}R(K)^2\int_0^1\int_0^1c(u,v)\ du\ dv+\frac{h_n^4\mu_2(K)^2}{4}\int_0^1\int_0^1[c_{uu}(u,v)+c_{vv}(u,v)]^2\ du\ dv+o((nh_n^2)^{-1}+h^4)\nonumber\\ =&
   (nh_n^2)^{-1}R(K)^2+\frac{h_n^4\mu_2(K)^2}{4}\beta+o((nh_n^2)^{-1}+h^4),
\end{align}
where \(\beta=\int_0^1\int_0^1[c_{uu}(u,v)+c_{vv}(u,v)]^2 \, du \, dv\). This leads to the asymptotic MISE
\begin{align}\label{AMISE}
    \text{AMISE}(h_n) = (nh_n^2)^{-1}R(K)^2 + \frac{h_n^4 \mu_2(K)^2}{4} \beta,
\end{align}
with \(\beta\) as defined above. The optimal bandwidth that minimizes the MISE is therefore given by
\begin{align}\label{optimal_mise}
    h_n^* = \left[ \frac{2 R(K)^2}{n \mu_2(K)^2 \beta} \right]^{1/6}.
\end{align}

This expression still relies on the unknown copula density $c(u,v)$. In practice, we select a parametric reference copula family and adjust the parameter to match the dependence strength observed in the data, for instance, by inverting Kendall's tau to estimate the unknown parameter in Frank's copula density. The optimal bandwidth $h_n^*$ is then determined numerically. It is important to note that this approach is feasible only if $R(K)^2$ and $\beta$ are finite.

\subsection{Cross-validation (CV)}
We now shift our focus to an alternative approach for estimating bandwidth. Notably, \( h_n^* \), defined in (\ref{optimal_mise}), is obtained by minimizing the AMISE, where estimates replace the unknown curvature components of \( c \). Alternatively, we may directly minimize the MISE using cross-validation approaches, as recommended by Scott and Terrell (1987) \cite{scott1987biased}. The CV approach involves using the sample twice: once to compute the KDE and again to assess its accuracy in estimating \( c(u, v) \). To avoid dependence on the same data for both computation and evaluation, the CV approach partitions the sample in a cross-validatory manner, ensuring that the data used to compute the KDE is excluded from its evaluation.

\subsubsection{Least squares cross-validation (LSCV)}

Least squares cross-validation (LSCV) is a statistical technique for evaluating a model’s predictive accuracy and optimizing model parameters by dividing data into training and validation subsets. This method is particularly useful in selecting parameters that prevent overfitting, ensuring the model generalizes well to new data. Rudemo (1982) \cite{rudemo1982empirical} discussed LSCV for selecting smoothing parameters in density estimation, while Bowman (1984) \cite{bowman1984alternative} introduced foundational approaches to LSCV for bandwidth selection. Wand and Jones (1995) \cite{wand1995kernel} provided an overview of LSCV for kernel smoothing, covering both theoretical and practical applications.

\textit{\bf Second-stage smoothing}\\
In the second stage of the estimation process, we employ the estimated marginal distributions to estimate the copula density \( c(u, v) \), where \( u = F(x) \) and \( v = G(y) \). Here, the Epanechnikov kernel, $K(x)=3/4(1-x^2)_+$, where the subscript ``$+$'' denotes the positive part. Our goodness-of-fit criterion between \( c(u, v) \) and \( \hat{c}(u, v) \) is the usual ISE, defined as follows:
\begin{align}\label{ISE}
    \text{ISE} = \int_0^1 \int_0^1 \left[ \hat{c}(u, v) - c(u, v) \right]^2 \, \ du\ dv.
\end{align}
The ISE measures the discrepancy between the true copula density and the estimated density over the unit square. By minimizing the expected value of the ISE, we obtain the MISE, which represents the average error in our estimation over repeated samples (see Rudemo (1982) \cite{rudemo1982empirical} and Bowman (1984) \cite{bowman1984alternative}). The MISE serves as a criterion for selecting an optimal bandwidth by minimizing the estimation error. Replacing \(\hat{c}(u,v)\) with the generalized estimator \(\hat{c}(u,v;h_n)\) in (\ref{ISE}) and expanding yields
\begin{align*}
    ISE(\hat{c}(\cdot;h_n)) = R(\hat{c}(\cdot;h_n)) - 2 \int_0^1 \int_0^1 \hat{c}(u,v;h_n)c(u,v) \,  du\ dv + \int_0^1 \int_0^1 c(u,v)^2 \,  du\ dv,
\end{align*}
where
\[
\hat{c}(u,v;h_n) = \frac{1}{n} \sum_{i=1}^n \sum_{l=1}^9 K_{h_n} \left( u - \hat{U}_{il} \right) K_{h_n} \left( v - \hat{V}_{il} \right)
\]
is the multiplicative bivariate kernel estimator. Here, \( \hat{U}_{il}  \) and \( \hat{V}_{il}  \) are as given in Definition \ref{definition1}, and \( R(g) \) is as defined in (\ref{l2norm}). The MISE is then given by
    \begin{align*}
  MISE(\hat{c}(\cdot;h_n)) =E\left[R(\hat{c}(\cdot;h_n))\right]-2E\left[\int_0^1\int_0^1\hat{c}(u,v;h_n)c(u,v)\ du\ dv\right]+E\left[\int_0^1\int_0^1 c(u,v)^2\ du\ dv\right].
\end{align*}
Since the last term is independent of \( h_n \), minimizing \( MISE(\hat{c}(\cdot; h_n)) \) is equivalent to minimizing 
\begin{align}\label{minimised}
 E\left[R(\hat{c}(\cdot;h_n))\right]-2E\left[\int_0^1\int_0^1 \hat{c}(u,v;h_n)c(u,v)\ du\ dv\right].   
\end{align}
This quantity is unknown but can be estimated unbiasedly as (Bowman (1984) \cite{bowman1984alternative} and Wand and Jones (1995) \cite{wand1995kernel}):
\begin{align}\label{LSCV}
    \text{LSCV}(h_n) := R(\hat{c}(\cdot; h_n)) - \frac{2}{n} \sum_{i=1}^n \sum_{l=1}^9 \hat{c}_{-i}(U_{il}, V_{il}; h_n),
\end{align}
where \begin{align*}
    \hat{c}_{-i}(u, v; h_n) = \frac{1}{n-1} \sum_{\substack{j=1 \\ j \neq i}}^n \sum_{l=1}^9 K_{h_n}(u - U_{jl}) K_{h_n}(v - V_{jl})
\end{align*} is the leave-one-out KDE, computed by excluding \( (U_{il}, V_{il}) \) from the sample.

The first term in (\ref{LSCV}) is unbiased by design. The second term results from estimating $\int\int \hat{c}(u,v)c(u,v) \, du\ dv$ using a Monte Carlo approximation based on the sample $\lbrace U_{il}, V_{il} \mid i = 1, \ldots, n; \, l = 1, \ldots, 9 \rbrace$, whose copula has density $c(u,v)$; specifically, this is achieved by replacing $c(u,v) \, \ du\ dv = dC(u,v)$ with $dC_n(u,v)$, where $C_n$ denotes the empirical copula.
This can be expressed as 
\begin{align*}
    \int\int \hat{c}(u,v;h_n) \, du \, dv \approx \frac{1}{n} \sum\limits_{i=1}^n \hat{c}(U_{il}, V_{il}; h_n),
\end{align*}
where, to reduce sample dependence, we substitute $\hat{c}(U_{il}, V_{il}; h_n)$ with $\hat{c}_{-i}(U_{il}, V_{il}; h_n)$. In this manner, we utilize the sample to estimate the integral involving $\hat{c}(\cdot; h_n)$, while for each $(U_{il}, V_{il})$, the kernel density estimate is computed from the remaining observations.

To demonstrate that (\ref{LSCV}) serves as an unbiased estimator of (\ref{minimised}), observe that the expectations of (\ref{minimised}) and $LSCV(h_n)$ correspond term by term, since \begin{align*}
    E \hat{c}_{-i}(u_i,v_i;h_n)&=E K_{h_{n}}(u_i-U_j)K_{h_{n}}(v_i-V_j)=E\int_0^1\int_0^1 K_{h_{n}}(u-U_j)K_{h_{n}}(v-V_j)c(u,v)\ du\ dv\\ &=E\hat{c}(u,v;h_n)\ du\ dv.
\end{align*}

We refer to (\ref{LSCV}) as an unbiased cross-validation criterion because its expectation satisfies
\begin{align*}
   E(\text{LSCV}(h_n)) = \text{MISE}(\hat{c}(\cdot; h_n)) - \int_0^1 \int_0^1 c(u, v)^2 \, du \, dv.
\end{align*}
This expression shows that the expected value of LSCV directly approximates the MISE, corrected by the integral term. In contrast, other theoretical expressions, such as AMISE, are only asymptotically unbiased, which means they approximate the MISE as \( n \to \infty \) but are biased for finite sample sizes.

The LSCV selector is then defined by \begin{align*}
    \hat{h}_{\text{LSCV}} := \arg\min_{h_n > 0}\text{LSCV}(h_n).
\end{align*}
To obtain $\hat{h}_{\text{LSCV}}$, numerical optimization is necessary. However, this process can be complicated by the LSCV function, which may contain multiple local minima, and its objective function can exhibit significant roughness, depending on $n$ and $c$. Consequently, optimization algorithms may sometimes converge to incorrect solutions. To avoid this, one can verify the solution by plotting $\text{LSCV}(h_n)$ over a range of $h_n$ values or by performing a search over a specified bandwidth grid. Hall (1983) \cite{hall1983large} and Stone (1984) \cite{stone1984asymptotically} demonstrated that the LSCV procedure yields a consistent sequence of smoothing parameters and is, in a specific sense, asymptotically optimal. 

Analyzing the MISE in \eqref{MISE} requires understanding the kernel moments in \eqref{moments dfn}, which is not immediately evident in \eqref{LSCV}. For large \( n \), the bandwidth becomes small enough that the probability mass outside the unit square tends to zero. Using this, Gijbels and Mielniczuk (1990) \cite{gijbels1990} derived the bias of the mirror reflection estimator (Theorem 3.2). Incorporating this concept, \eqref{LSCV} can be reformulated as:
 \begin{align}\label{2nd form of LSCV}
    \text{LSCV}(h)=\frac{R(K)^2}{nh^2} +& \sum_{i=1}^n\sum_{\substack{j=1\\i\neq j}}^n\bigg[ \frac{1}{n^2h^4}\int_0^1K\left(\frac{u-\hat{U_i}}{h_n}\right)K\left(\frac{u-\hat{U_j}}{h_n}\right)\ du \int_0^1K\left(\frac{v-\hat{V_i}}{h_n}\right)K\left(\frac{v-\hat{V_j}}{h_n}\right)\ dv\nonumber\\
    &\qquad -\frac{2}{n(n-1)h^2}K\left(\frac{\hat{U_i}-\hat{U_j}}{h_n}\right)K\left(\frac{\hat{V_i}-\hat{V_j}}{h_n}\right)\bigg].
\end{align}
An instructive exercise shows that for \( p \) even, the expectation of \eqref{2nd form of LSCV} equals \eqref{AMISE} minus the constant \( R(c(u,v)) \). Under the kernel properties in Assumptions \ref{assumptions}, we define:
\begin{align}\label{gammatautau}
    \gamma(\tau_1, \tau_2) = \int\int K(w_1)K(w_1+\tau_1)K(w_2)K(w_2+\tau_2) \, dw_1 \, dw_2 - 2K(\tau_1)K(\tau_2),
\end{align}
and let \begin{align}
    \tau_{ij,1}=(\hat{U}_i-\hat{U}_j)/h_n\ \text{and}\ \tau_{ij,2}=(\hat{V}_i-\hat{V}_j)/h_n,
\end{align}
then \eqref{2nd form of LSCV} becomes (replacing $n-1$ by $n$)\begin{align}\label{new LSCV}
    LSCV(h_n)=\frac{R(K)^2}{nh_n^2} + \frac{2}{n^2h_n^2}\mathop{\sum\sum}_{ \ i< j \ }\gamma(\tau_{ij,1},\tau_{ij,2}).
\end{align}
The following theorem provides the mean of the function \eqref{new LSCV} for fixed $h_n$.
\begin{theorem}
 For the LSCV kernel criterion \eqref{new LSCV}\begin{align}\label{AExpectation}
     E[LSCV(h_n)]=AMISE(h_n)-R(c)+O(n^{-1}).
 \end{align}   
\end{theorem}
\begin{proof}
Recall that $K$ is a symmetric function with support on $[-1,1]$. Using this property, we define $\gamma_+(\tau_1,\tau_2)$ and $\gamma_-(\tau_1,\tau_2)$ to represent $\gamma(\tau_1,\tau_2)$ as given in \eqref{gammatautau}, over the intervals $[0,2]^2$ and $[-2,0]^2$, respectively. For $0 \leq \tau_1, \tau_2 \leq 2$, $\gamma_+(\tau_1,\tau_2)$ is given by:
\begin{align}\label{gamma+}
\gamma_+(\tau_1,\tau_2) \equiv \int_{-1}^{1-\tau_1} \int_{-1}^{1-\tau_2} K(w_1)K(w_1+\tau_1)K(w_2)K(w_2+\tau_2)\ dw_1\ dw_2 - 2K(\tau_1)K(\tau_2).
\end{align}
Similarly, for $-2 \leq \tau_1, \tau_2 \leq 0$, $\gamma_-(\tau_1,\tau_2)$ is defined analogously to \eqref{gamma+}, with the integration limits for both $\tau_1$ and $\tau_2$ replaced by $-1-\tau$ and $1$. Since $K$ is symmetric, it follows that $\gamma$ is also symmetric.

With these definitions in place, we can now evaluate:
\begin{align}
E[\gamma(\tau_{ij,1},\tau_{ij,2})]= &\int_0^1\int_0^1 c(u,v)\bigg[ \int_{u-2h_n}^u\int_{v-2h_n}^v\gamma_+\left(  \frac{u-u_1}{h_n}, \frac{v-v_1}{h_n}\right)c(u_1,v_1)\ du_1\ dv_1 +\nonumber\\
&\qquad\int_u^{u+2h_n}\int_v^{v+2h_n}\gamma_-\left(  \frac{u-u_1}{h_n}, \frac{v-v_1}{h_n}\right)c(u_1,v_1)\ du_1\ dv_1 \bigg]\ du\ dv\nonumber\\
&=h_n^2\int_0^1\int_0^1 c(u,v)\bigg[ \int_0^2\int_0^2 \gamma_+(\tau_{1},\tau_{2})\lbrace c(u-\tau_{1}h_n,v-\tau_{2}h_n) + \nonumber\\ &\qquad c(u+\tau_{1}h_n,v+\tau_{2}h_n)\rbrace\ d\tau_1\ d\tau_2 \bigg]\ du\ dv\nonumber\\
&=2 h_n^2\int_0^1\int_0^1 c(u,v)\bigg[ \int_0^2\int_0^2 \gamma_+(\tau_{1},\tau_{2})\bigg( c(u, v) + \sum_{|\kappa| = 2} \frac{h_n^{|\kappa|} \tau_1^{\kappa_1} \tau_2^{\kappa_2}}{\kappa_1! \kappa_2!} \frac{\partial^{|\kappa|} c(u,v)}{\partial u^{\kappa_1} \partial v^{\kappa_2}} +\nonumber\\
&\qquad \sum_{|\kappa| = 4} \frac{h_n^{|\kappa|} \tau_1^{\kappa_1} \tau_2^{\kappa_2}}{\kappa_1! \kappa_2!} \frac{\partial^{|\kappa|} c(u,v)}{\partial u^{\kappa_1} \partial v^{\kappa_2}} \bigg) \ d\tau_1\ d\tau_2 \bigg]\ du\ dv,
\end{align}
where $|\kappa|=\kappa_1+\kappa_2$ is the total order of differentiation. Now, for $\kappa_1$ and $\kappa_2$ even\begin{align}\label{meansgamma}
    \int_0^2\int_0^2 \tau_1^{\kappa_1}\tau_2^{\kappa_2}\gamma_+(\tau_{1},\tau_{2})\ d\tau_1\ d\tau_2=&\frac{1}{2}\int_{-1}^1 K(w_1)\int_{-1}^1 (s_1-w_1)^{\kappa_1}K(s_1)\ ds_1\ dw_1\times \nonumber\\ &\qquad\int_{-1}^1 K(w_2)\int_{-1}^1 (s_2-w_2)^{\kappa_2}K(s_2)\ ds_2\ dw_2-\mu_{\kappa_1}\mu_{\kappa_2},
\end{align}
where we have used the substitution $s_1=w_1+\tau_1$ and $s_2=w_2+\tau_2$. The expression \eqref{meansgamma} equals to $-1/2, 0, \mu_2^2$ and $3\mu_2^2$ for $\lbrace \kappa_1,\kappa_2\rbrace=\lbrace 0,0 \rbrace, \lbrace 0,2 \rbrace, \lbrace 2,2 \rbrace, \lbrace 0,4 \rbrace$, respectively. Hence \begin{align}\label{exptauijs}
 E[\gamma(\tau_{ij,1},\tau_{ij,2})]=   & h_n^2 \int_0^1 \int_0^1  \bigg[ -c(u, v)^2 + \frac{\mu_2^2h_n^4}{4} \bigg( c(u, v)\frac{\partial^4 c(u,v)}{\partial u^4} +2c(u, v)\frac{\partial^4 c(u,v)}{\partial u^2 \partial v^2} + \nonumber\\
 &\qquad c(u, v)\frac{\partial^4 c(u,v)}{\partial v^4} \bigg) \bigg] \ du\ dv\nonumber\\
 &= h_n^2\bigg[ -R(c(u,v))+\frac{\mu_2^2h_n^4}{4}\int_0^1 \int_0^1\left[ c_{uu}(u,v)+c_{vv}(u,v) \right]^2 \bigg],
\end{align}
where we have used integration by parts. Thus \eqref{AExpectation} follows from \eqref{exptauijs}, \eqref{new LSCV} and \eqref{MISE}.
\end{proof}

The roughness of $\hat{c}(u,v)$, represented by $R(\hat{c}(\cdot;h_n))$ in (\ref{LSCV}), when using the Epanechnikov kernel can be described in terms of the convolution of the kernel $K$ with itself. Before presenting this result, we first state the following lemma, which will be instrumental in the proof.

\begin{lemma}\label{lemmaKint}
Let \( K \) be a symmetric kernel. Then
\begin{align*}
    \int K\left( \frac{u - U_i}{h} \right) K\left( \frac{u - U_j}{h} \right) \, du = h \left( K * K \right) \left( \frac{U_i - U_j}{h} \right),
\end{align*}
where \( K * K \) denotes the convolution of \( K \) with itself.
\end{lemma}
\begin{proof}
To prove Lemma \ref{lemmaKint}, we start by applying the substitution \( u = hw \), which implies \( du = h \, dw \). Substituting this into the integral, we obtain:
\begin{align}\label{Kint}
    \int K\left( \frac{u - U_i}{h} \right) K\left( \frac{u - U_j}{h} \right) \, du 
    = h \int K\left( \frac{U_i}{h} - w \right) K\left( w - \frac{U_j}{h} \right) \, dw.
\end{align}

The expression on the right side of equation \eqref{Kint} can be interpreted as a product of two shifted kernels centered at \( U_i/h \) and \( U_j/h \), respectively. For clarity, let us denote these shifted kernels as \(
K_{U_i}(w) = K\left( U_i/h - w \right) \quad \text{and} \quad K_{U_j}(w) = K\left( w - U_j/h \right).
\)
Thus, we can rewrite equation \eqref{Kint} as:
\begin{align}\label{Kint1}
    h \int K\left( \frac{U_i}{h} - w \right) K\left( w - \frac{U_j}{h} \right) \, dw 
    = h \int K_{U_i}(w) K_{U_j}(w) \, dw.
\end{align}

Now, to further simplify the right side of \eqref{Kint1}, let us make the substitution \( w_1 = w - U_j/h \), so that \( w = w_1 + U_j/h \) and \( dw = dw_1 \). Under this substitution \(
K_{U_j}(w) = K(w_1) \quad \text{and} \quad K_{U_i}(w) = K\left( (U_i - U_j)/h - w_1 \right).
\)
Thus, we can rewrite \eqref{Kint1} as:
\[
h \int K\left( \frac{U_i - U_j}{h} - w_1 \right) K(w_1) \, dw_1 = h \left( K * K \right) \left( \frac{U_i - U_j}{h} \right),
\]
as required.
\end{proof}

For the Epanechnikov kernel, the convolution $(K * K)(x)=\int K(t)K(x-t)\ dt$ implies that $-1\leq t\leq 1$ and $-1\leq x-t\leq 1$. These two inequalities imply that $x-1\leq t\leq x+1$ and $x\in [-2,2].$ Thus the integral of the convolution of the Epanechnikov kernel by itself can be broken into three cases based on $x$ as\begin{align}\label{convof Epi}
   (K * K)(x)=\begin{cases}
       \int_{-1}^{x+1}\frac{9}{16}(1-t^2)(1-(x-t)^2)dt, & -2\leq x\leq -1,\\
       \int_{-1}^{1}\frac{9}{16}(1-t^2)(1-(x-t)^2)dt, & -1\leq x\leq 1,\\
       \int_{x-1}^{1}\frac{9}{16}(1-t^2)(1-(x-t)^2)dt, & 1\leq x\leq 2,\\
       0, & \text{otherwise.}
   \end{cases} 
\end{align}

\begin{prop}\label{Rcprop}
 Let \( K \) be a univariate kernel. The roughness of the estimator \( \hat{c}(u,v) \) given in \eqref{mult kernel}, when using the Epanechnikov kernel, is of the form
 \begin{align}\label{R(chat)1}
     R(\hat{c}) = \frac{81}{256n^2h_n^2} \sum_{i=1}^n \sum_{j=1}^n \sum_{l=1}^9 \sum_{m=1}^9 (K * K)(\xi_{ijlm}) \left( K * K \right)(\zeta_{ijlm}),
 \end{align} 
 where \( \xi_{ijlm} = (U_{il} - U_{jm}) / h_n \), \( \zeta_{ijlm} = (V_{il} - V_{jm}) / h_n \), and \( (K * K)(x) \) is the convolution of the Epanechnikov kernel with itself.
\end{prop}

\begin{proof}
Using Definition \ref{definition1},
\begin{align}\label{roughness}
    R(\hat{c})&= \frac{1}{n^2}\int\int \left(\sum\limits_{i=1}^n\sum\limits_{l=1}^9 K_{h_n}(u-U_{il})K_{h_n}(v-V_{il})\right)^2\ du\ dv\nonumber\\
    &=\frac{1}{n^2} \sum\limits_{i=1}^n\sum\limits_{l=1}^9\sum\limits_{j=1}^n\sum\limits_{m=1}^9 \int\int K_{h_n}(u-U_{il})K_{h_n}(v-V_{il})K_{h_n}(u-U_{jm})K_{h_n}(v-V_{jm})\ du\ dv
\end{align}
Using the Epanechnikov kernel, then (\ref{roughness}) can be written as
\begin{align}\label{roughness2}
 R(\hat{c})&=\frac{81}{256n^2h^4}  \sum\limits_{i=1}^n\sum\limits_{l=1}^9\sum\limits_{j=1}^n\sum\limits_{m=1}^9 \int\int\bigg\{ \bigg( 1-\bigg(\frac{u-U_{il}}{h_n}\bigg)^2 \bigg)\bigg( 1-\bigg(\frac{v-V_{il}}{h_n}\bigg)^2 \bigg) \bigg( 1-\bigg(\frac{u-U_{jm}}{h_n}\bigg)^2 \bigg)\times\nonumber\\ &\qquad\qquad\bigg( 1-\bigg(\frac{v-V_{jm}}{h_n}\bigg)^2 \bigg)\bigg\} \ du\ dv\nonumber\\
 &=\frac{81}{256n^2h_n^4}  \sum\limits_{i=1}^n\sum\limits_{l=1}^9\sum\limits_{j=1}^n\sum\limits_{m=1}^9 \int \bigg( 1-\bigg(\frac{u-U_{il}}{h_n}\bigg)^2 \bigg)\bigg( 1-\bigg(\frac{u-U_{jm}}{h_n}\bigg)^2 \bigg)\ du\times \nonumber\\ &\qquad\qquad \int \bigg( 1-\bigg(\frac{v-V_{il}}{h_n}\bigg)^2 \bigg)\bigg( 1-\bigg(\frac{v-V_{jm}}{h_n}\bigg)^2 \bigg)\ dv,\nonumber\\
 &=\frac{81}{256n^2h_n^2}\sum\limits_{i=1}^n\sum\limits_{j=1}^n\sum\limits_{l=1}^9\sum\limits_{m=1}^9( K* K)(\xi_{ijlm})\left( K* K \right)(\zeta_{ijlm}),
\end{align}
where we have used Lemma \ref{lemmaKint} to arrive at the last line of (\ref{roughness2}), $(K* K)$ is the convolution of the Epanechnikov kernel by itself defined in (\ref{convof Epi}), and $\xi_{ijlm}=(U_{il}-U_{jm})/h_n$, $\zeta_{ijlm}=(V_{il}-V_{jm})/h_n$. 
\end{proof}

\subsubsection{Biased cross-validation (BCV)}

Unlike the LSCV, which relies solely on the MISE, the Biased Cross-Validation (BCV) approach is a hybrid that combines both plug-in and cross-validation techniques. To introduce the bivariate case, we begin by examining the univariate BCV cases discussed in existing literature. The work by Scott and Terrell (1987) \cite{scott1987biased} proposed minimizing the AMISE:
\begin{align}\label{AMISE2}
    \text{AMISE}(h) = \frac{R(K)}{nh} + \frac{h^4 \mu_2(K)^2}{4} R\big(f''(x)\big),
\end{align}
where $h$ represents the bandwidth for the univariate kernel $K$, $R(K)$ and $\mu_2(K)$ are as defined in \eqref{l2norm} and \eqref{moments dfn}, respectively, and $f''(x)$ is the second derivative of the underlying univariate density. For more details on univariate kernel density estimation, see Scott and Terrell (1987) \cite{scott1987biased} and Wand and Jones (1995) \cite{wand1995kernel}. The only unknown term in the AMISE expression in \eqref{AMISE2} is the second derivative of the underlying density.

Scott and Terrel (1987) \cite{scott1987biased} proposed an estimate of \( R(f^{''}(x)) \) of the form $\hat{R}^1(f^{''}) = R(\hat{f}^{''}) - R(K^{''})/nh^5,$
where \( \hat{f}^{''} \) is the second derivative of the univariate KDE, \( K \). They gave the corresponding AMISE estimate as
\begin{align*}
    BCV^1(h) = \frac{R(K)}{nh} + \frac{h^4 \mu_2(K)^2}{4} \hat{R}^1(f^{''}).
\end{align*}

Hall and Marron (1987) \cite{hall1987estimation} investigated the precision of estimators for the quantity $R(f^{''}(x))$ and derived an estimator using the identity
\begin{align*}
    R(f^{''}) = \int f^{''}(x)^2 \, dx = \int f^{(iv)}(x) f(x) \, dx = E[f^{''}(x)],
\end{align*}
while also discussing its implications for bandwidth selection. Utilizing this identity, they demonstrated that $\hat{R}^2(f^{''}) = \frac{1}{n} \sum_{i=1}^n \hat{f}_{-i}^{(iv)}(x_i)$, leading to the criterion function 

\begin{align*}
    BCV^2(h) = \frac{R(K)}{nh} + \frac{h^4 \mu_2(K)^2}{4} \hat{R}^2(f^{''}).
\end{align*}

In the univariate case, the estimator $\hat{R}^2$ has a higher variance than $\hat{R}^1$ (Sain \textit{et al.}, 1994 \cite{sain1994cross}). However, this increased variance is accompanied by a reduction in bias. Importantly, $\hat{R}^2$ offers greater ease of extension to multivariate cases (Sain \textit{et al.}, 1994 \cite{sain1994cross}), making it an appealing choice for our analysis due to its implementational simplicity and improved bias characteristics.

Consider the AMISE given in (\ref{AMISE}). We can rewrite this as \begin{align}\label{BCVAMISE}
    \text{AMISE}(h) = &\ \frac{R(K)^2}{n h_n^2} + \frac{h_n^4 \mu_2(K)^2}{4} \Bigg( \int_0^1 \int_0^1 c_{uu}(u,v)^2 \, du \, dv \nonumber\\
    &\ + \int_0^1 \int_0^1 c_{vv}(u,v)^2 \, du \, dv + 2 \int_0^1 \int_0^1 c_{u}(u,v) c_{v}(u,v) \, du \, dv \Bigg).
\end{align}
where $K, R(K)$ and $\mu_2(K)$ are as used in (\ref{AMISE}) and $\beta$ in  (\ref{AMISE}) has been expanded to obtain the factor in parenthesis in the second term. 

Using the bivariate form of $\hat{R}^2$, we can estimate the roughness of the first second-order partial derivative in the AMISE expression as follows:
\begin{align*}
    \int_0^1 \int_0^1 c_{uu}(u,v)^2 \, du \, dv
    =\frac{\partial^4c(u,v)}{\partial u^4}\ c(u,v)du\ dv=E\left[\frac{\partial^4c(u,v)}{\partial u^4} \right].
\end{align*}
The second term in the AMISE expression can be dealt with similarly while the third is estimated  as\begin{align*}
 \int_0^1 \int_0^1 c_{uu}(u,v) c_{vv}(u,v) \, du \, dv= \frac{\partial^4c(u,v)}{\partial u^2\ \partial v^2}\ c(u,v)du\ dv=  E\left[\frac{\partial^4c(u,v)}{\partial u^2\ \partial v^2} \right]
\end{align*} 
These expressions can be estimated by \begin{align}\label{BCV11}
    n^{-1}\sum\limits_{i=1}^n \frac{\partial^4\hat{c}_{-i}(U_i,V_i)}{\partial u^4} 
\end{align}
and
\begin{align}\label{BCV12}
    n^{-1}\sum\limits_{i=1}^n \frac{\partial^4\hat{c}_{-i}(U_i,V_i)}{\partial u^2\ \partial v^2}. 
\end{align}
Substituting \eqref{BCV11} and \eqref{BCV12} into \eqref{BCVAMISE} we obtain the BCV function. Typically, a kernel with compact support, such as the Epanechnikov kernel, may not suffice for higher-order derivatives due to discontinuities at the boundaries. A Gaussian kernel or other infinitely differentiable kernels are more appropriate choices.

\section{Simulation Study}\label{simulation}
One of the most insightful aspects of working with copulas is the ability to create exploratory plots that visually represent the dependence structure. In this section, we present visualizations of various aspects of the true density and the mirror-reflection estimator, based on simulated data \((n = 500)\) generated from a Frank copula with parameter \(\theta = 5\) (\(\text{Kendall's } \tau \approx 0.46\)). All computations are coded on \texttt{R Studio} software. Note that the Frank copula has a bounded density. The simple mirror-reflection estimator fails to capture the tail behavior of unbounded densities, such as those of the Gaussian copula or any copula exhibiting tail dependence.

Figure~\ref{fig1}(a) below shows the scatterplot of the original data generated from the Frank copula model. In Figure~\ref{fig1}(b), the data is transformed using the inverse of the Gaussian cumulative distribution function (CDF), which highlights a positive dependence between the variables.

\begin{figure}[ht!]
    \centering
    \begin{subfigure}[b]{0.35\textwidth}
   \includegraphics[width=\textwidth, height=5cm]{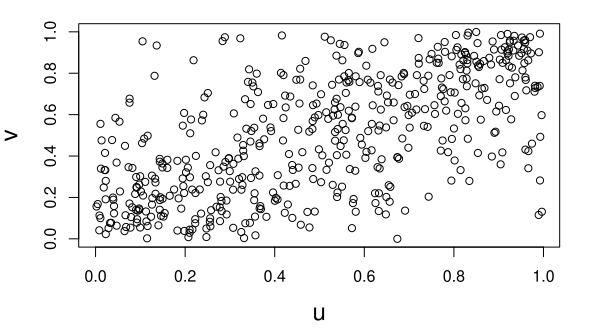}
        \caption{}
        \label{fig1a}
    \end{subfigure}
   \hspace{0.05\textwidth}
    \begin{subfigure}[b]{0.35\textwidth}
        \includegraphics[width=\textwidth, height=5cm]{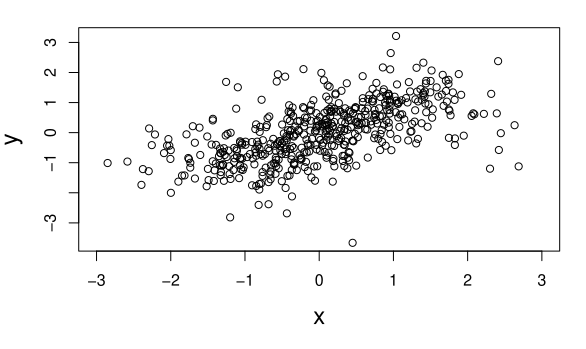}
        \caption{}
        \label{fig1b}
    \end{subfigure}
    \caption{(a) Scatterplot of original data from the Frank copula sample and (b) scatterplot of the transformed sample.}
    \label{fig1}
\end{figure}

To address the bandwidth selection problem, we considered various sample sizes and calculated the bandwidths that optimize the AMISE and minimize the LSCV\((h)\) criterion, respectively. Table~\ref{table:1} presents a comparison of the optimal bandwidths (bw), \(h^*_n\) and \(\hat{h}_{LSCV}\), along with the ISE for different sample sizes. The rule-of-thumb selector has a convergence rate much faster than the LSCV selector.

\begin{table}[ht]
    \centering
    \renewcommand{\arraystretch}{1.6}
    \begin{tabular}{c|cc|cc|cc|cc}
        \hline
        & \multicolumn{2}{c|}{$n=100$} & \multicolumn{2}{c|}{$n=200$} & \multicolumn{2}{c|}{$n=500$} & \multicolumn{2}{c}{$n=1000$} \\
        \hline
        & bw & ISE & bw & ISE & bw & ISE & bw & ISE \\
        \hline
        $h^*_n$ & 0.239 & 1.138 & 0.213 & 1.030 & 0.183 & 0.917 & 0.163 & 0.652 \\
        $h_{\text{LSCV}}$ & 0.244 & 1.227 & 0.238 & 1.094 & 0.196 & 1.052 & 0.199 & 1.007 \\
        \hline
    \end{tabular}
    \vspace{1mm} 
    \caption{Comparison of $h^*_n$ and $\hat{h}_{\text{LSCV}}$ for different sample sizes ($n$) in terms of bandwidth (bw) and ISE.}
    \label{table:1}
\end{table}

Marginal normal contour plots of the true density and the mirror-reflection estimator are shown in Figure \ref{fig1(1)}. Figure \ref{fig1(1)} (a) provides the contour plots of the true density, while Figure \ref{fig1(1)} (b) displays the mirror-reflection estimator based on the simulated data.  Bandwidths for the kernel density estimates were selected on the bases of AMISE-optimality. Examining the contour plots in Figure \ref{fig1(1)}, we conclude that the bandwidth selection rules are functioning appropriately for our finite samples. Using smaller bandwidths(resp., sample sizes) would result in wiggly estimates, which would make the visualizations a little unpleasant.

\begin{figure}[ht!]
    \centering
    \begin{subfigure}[b]{0.45\textwidth}
        \includegraphics[width=1\textwidth, height=5cm]{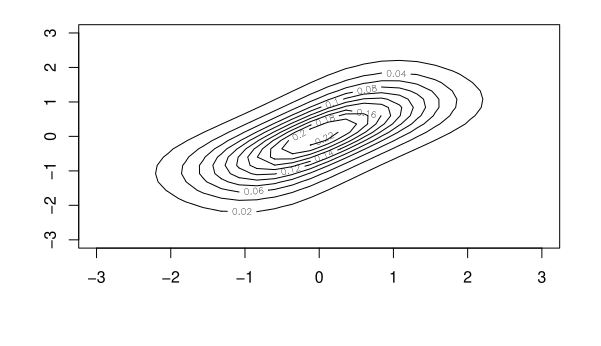}
        \caption{}
        \label{fig1(1)a}
    \end{subfigure}
    \hspace{0.05\textwidth}
    \begin{subfigure}[b]{0.45\textwidth}
        \includegraphics[width=.9\textwidth, height=5cm]{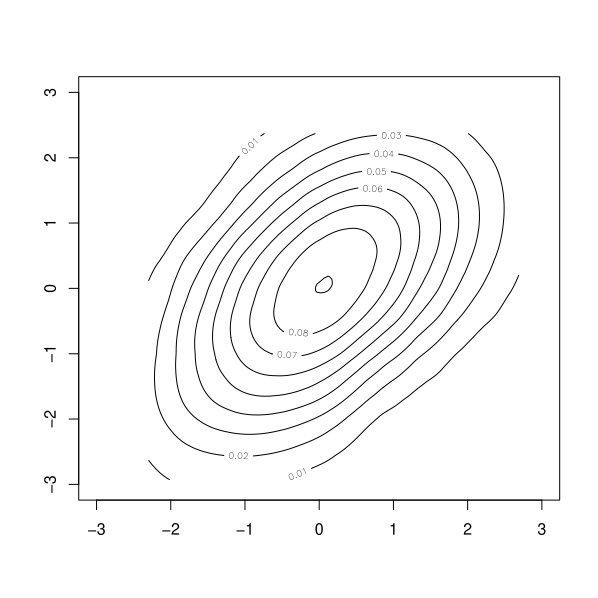}
        \caption{}
        \label{fig1(1)b}
    \end{subfigure}
    \caption{(a) Contour plot of the true density; (b) Contour plot of the mirror reflection estimates on simulated data \textcolor{red}(n=500) of a Frank copula. Bandwidths are selected based on the LSCV criterion.}
    \label{fig1(1)}
\end{figure}

Plot (a) in Figure \ref{fig1(3)} illustrates the mirror-reflection estimate of the density of the transformed data. The perspective plots in (b) shows the estimated density based on the simulated data of the Frank copula. Figure \ref{fig1(3)} (b) reveals some deviation of the estimated density from the true density (see the true Frank copula density plot in Nagler 2014 \cite{nagler2014kernel}), mainly due to the undersmoothing in the central region of the estimates. However, our bandwidth selection approach strikes a favorable balance between bias and variance, resulting in plots that closely align with the true density even with smaller samples.

\begin{figure}[ht!]
    \centering
    \begin{subfigure}[b]{0.35\textwidth}
        \includegraphics[width=\textwidth, height=4.5cm]{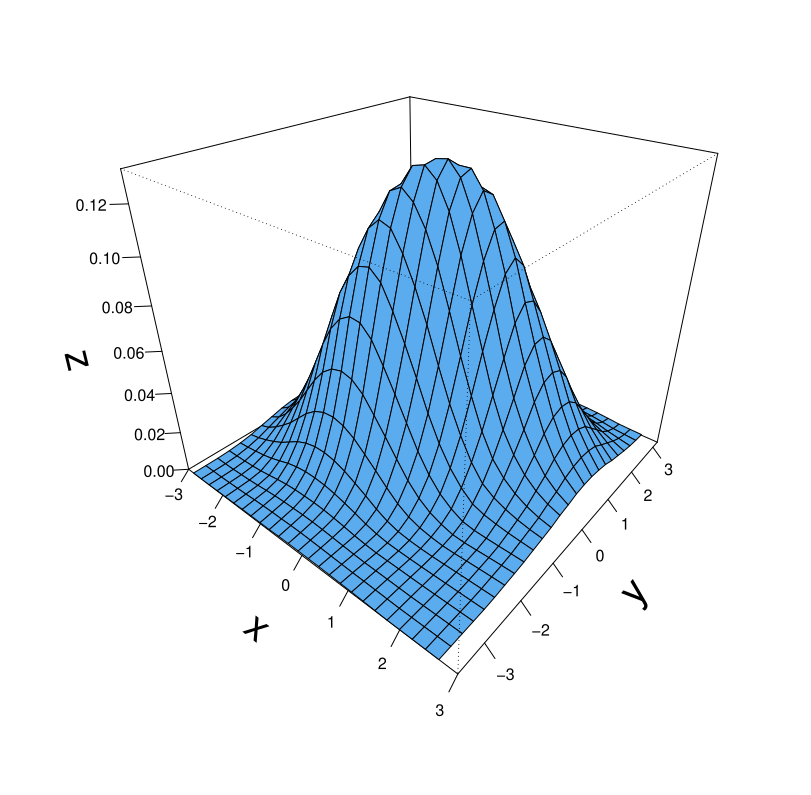}
        \caption{}
        \label{fig1(3)a}
    \end{subfigure}
    \hspace{0.05\textwidth} 
    \begin{subfigure}[b]{0.35\textwidth}
        \includegraphics[width=\textwidth, height=4.5cm]{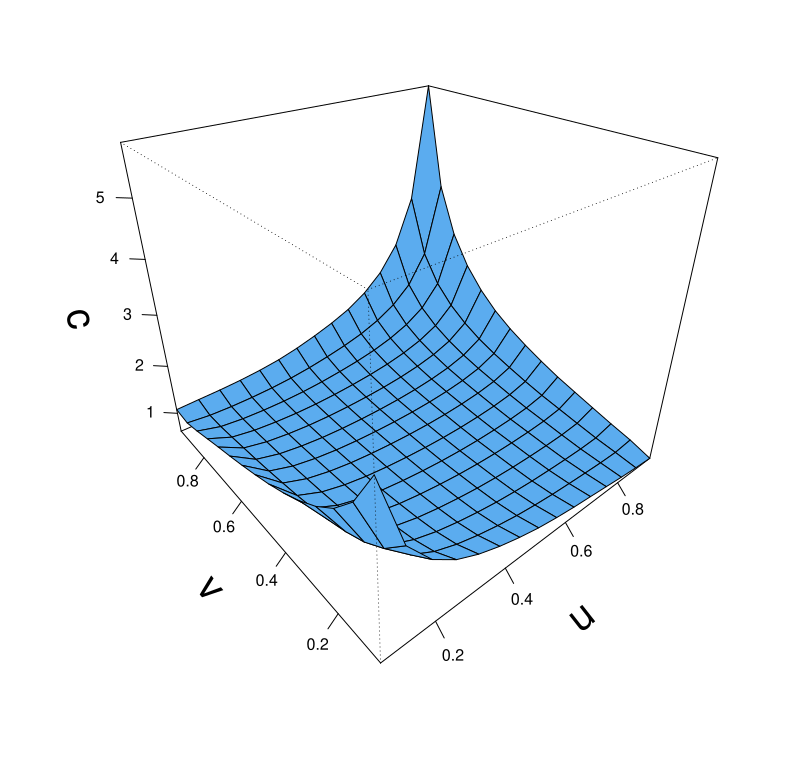}
        \caption{}
        \label{fig1(3)c}
    \end{subfigure}
    \caption{ (a) Mirror reflection estimate for the Frank copula transformed sample; (b) perspective plot for the true density of the Frank copula, and (c) perspective plot of the mirror-reflection kernel density estimate.}
    \label{fig1(3)}
\end{figure}

\section{Real-life data Application}\label{real-life}
In the practical application part of our study, we used the Wisconsin Breast Cancer Diagnostic Dataset (WBCDD) for analysis. This data set was derived from the examination of cell nuclei characteristics in 569 images obtained by Fine Needle Aspiration (FNA) of breast masses. Each sample in the dataset is categorized as either `Benign' or `Malignant', corresponding to noncancerous and cancerous cases, respectively.

The data set was compiled through the efforts of Dr. William Wolberg of Wisconsin Hospital, who provided the breast mass images, and William Nick Street from the Department of Computer Sciences, University of Wisconsin, who digitized the images in November 1995 (O. Mangasarian and W. Wolberg \cite{mangasarian1990cancer}). It contains 569 instances with 30 numeric features, in addition to an ID column and a class label. Among these samples, 357 are benign, while 212 are malignant. Notably, the data set does not have any missing values.

For this particular study, we chose two variables: the mean radius (column 2) and the mean concavity (column 8). The first step involved empirical estimation of the marginal distributions of these variables. This was achieved by computing their empirical cumulative distribution functions (ECDFs), \(\hat{F}_1\) and \(\hat{F}_2\), based on the observed data.  Specifically, for an observation \((X_i, Y_i)\), the transformed values were calculated as \((\hat{U}_i, \hat{V}_i)\), where \(\hat{U}_i = \hat{F}_1(X_i), \quad \hat{V}_i = \hat{F}_2(Y_i)\). This transformation effectively maps the data to the unit square \([0, 1]^2\) while preserving the dependence structure between the variables. Once the data was transformed to uniform margins, a scatter plot was created to visualize the relationships between the variables in their new scale. This visualization offered insights into the structure of dependence present in the dataset, even before any further transformations. 

Figure \ref{fig2} (a) shows a scatter plot of the data for the two variables: mean radius (column 2) and mean concavity (column 8) from the Wisconsin Cancer dataset, after transformation using their marginal empirical distribution functions. Figure \ref{fig2} (b) displays a scatter plot of the same data after transformation to the standard normal distribution. The scatter plot in (b) shows that the data have neither lower nor upper tail dependence.

\begin{figure}[ht!]
    \centering
    \begin{subfigure}[b]{0.40\textwidth}
        \centering
        \includegraphics[width=\textwidth, height=6cm]{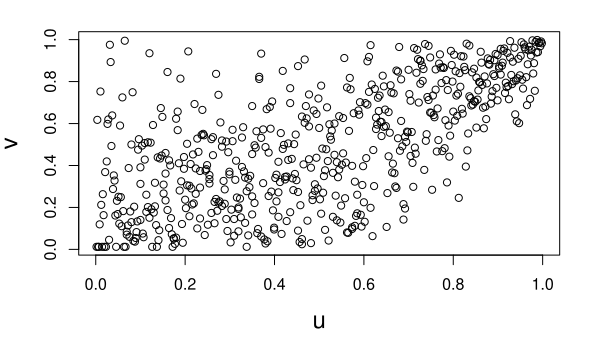}
        \caption{}
        \label{fig2a}
    \end{subfigure}
    \hspace{0.05\textwidth}
    \begin{subfigure}[b]{0.40\textwidth}
        \centering
        \includegraphics[width=\textwidth, height=6cm]{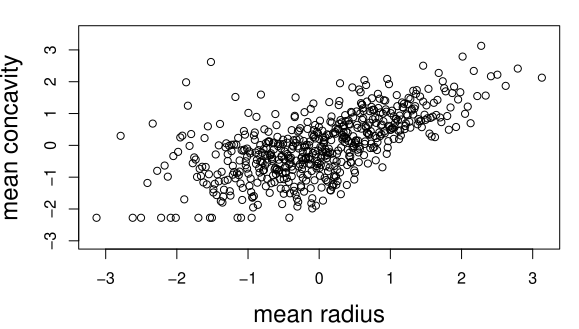}
        \caption{}
        \label{fig2b}
    \end{subfigure}
    \caption{(a) Scatter plot of data after transformation by its marginal empirical distribution function; (b) scatter plot of the transformed data with standard normal margins.}
    \label{fig2}
\end{figure}
The dependence structure of the WBCDD was then depicted using a contour plot. Figure \ref{fig2(1)} (a) shows contour plot of the copula density combined with standard normal margins, (b) the kernel density estimate for transformed sample (with standard normal margins),  (c) a surface or perspective plot of the copula density. The bandwidths were selected on the basis of the optimality of the LSCV criterion.

 \begin{figure}[ht!]
    \centering
    \begin{subfigure}[b]{0.32\textwidth}
        \centering
        \includegraphics[width=\textwidth, height=5cm]{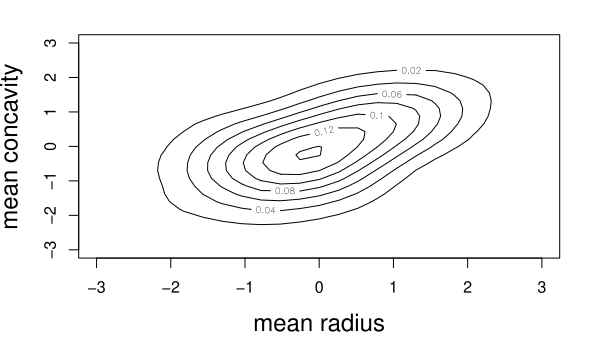}
        \caption{}
        \label{fig2(1)a}
    \end{subfigure}
    \begin{subfigure}[b]{0.32\textwidth}
        \centering
        \includegraphics[width=\textwidth, height=5cm]{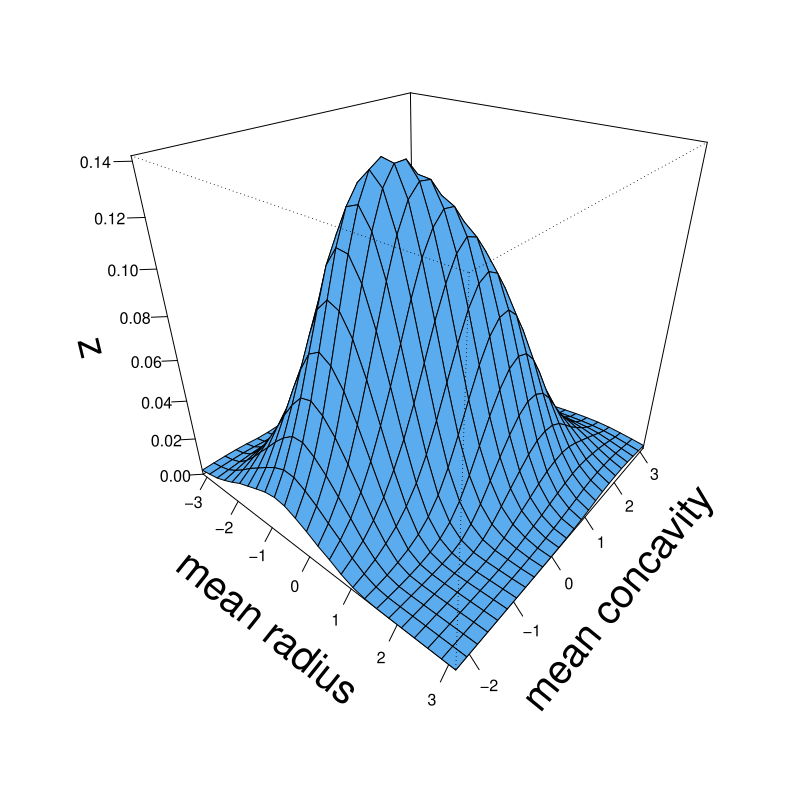}
        \caption{}
        \label{fig2(1)b}
    \end{subfigure}
    \begin{subfigure}[b]{0.32\textwidth}
        \centering
        \includegraphics[width=\textwidth, height=5cm]{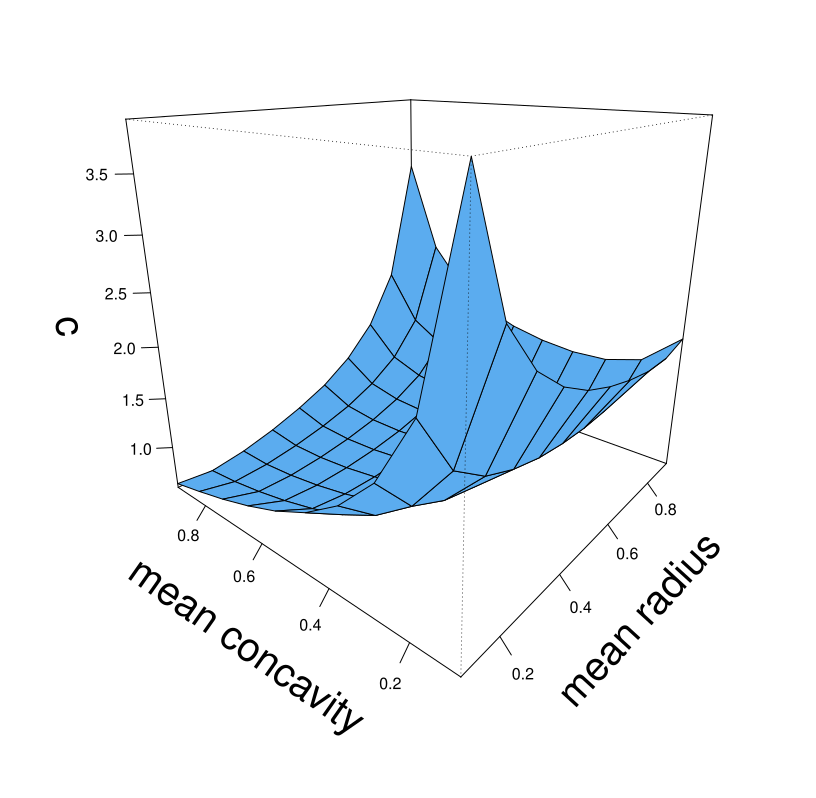}
        \caption{}
        \label{fig2(1)c}
    \end{subfigure}
    \caption{Exploratory visualizations of the WBCDD data and copula density: (a) Contour plot of the copula density combined with standard normal margins; (b) kernel density estimate for transformed sample (with standard normal margins); and (c) surface plot of the copula density.}
    \label{fig2(1)}
\end{figure}

\begin{figure}[ht!]
    \centering
{\includegraphics[width=0.44\textwidth, height=5cm]{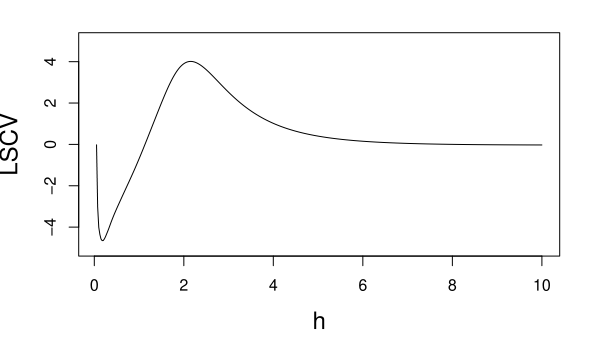}}
   \caption{LSCV curve for the mirror-reflection estimator applied to the WBCDD, showing the optimal smoothing parameter $h_{\text{LSCV}} = 0.031$.}
\label{fig3}
\end{figure}

In Figure~\ref{fig3}, we provided a plot of the LSCV curve for our mirror-reflection estimator for the WBCDD. The optimal LSCV smoothing parameter was found to be $h_{LSCV}=0.031$.

\section{Conclusion}
Methods for copula density estimation have been explored in the literature but remain insufficiently investigated, primarily due to the challenge of unbounded copula densities and their derivatives at the boundaries of \([0,1]^2\).  

In this study, we focused on estimating copula densities that are bounded within the unit square. This assumption facilitated a straightforward derivation of the asymptotic bias using Taylor expansion techniques. We introduced two kernel smoothing methods and demonstrated that the rule-of-thumb approach provides superior bandwidth selection, particularly in cases where the margins have unbounded support.  

We also proved a theorem establishing that the LSCV kernel criterion has an asymptotic expectation equal to the AMISE minus the roughness of the copula density. Simulation results showed that AMISE-optimal bandwidth selection is preferable in scenarios where the margins have unbounded support, even when the kernel used for estimation is bounded. Finally, for any copula to be applied to real-world data, it is essential to ensure that the data's distribution aligns with the theoretical properties of the chosen copula.

\end{document}